\documentclass[a4paper,UKenglish,cleveref, autoref]{lipics-v2019}
\usepackage[utf8]{inputenc}

\usepackage{amsmath,amsfonts}
\usepackage{mathtools}
\usepackage{stmaryrd}
\usepackage{bm}
\usepackage{tikz}
\usetikzlibrary{automata,petri,positioning,calc}
\usepackage{colortbl}
\usepackage{multirow}
\usepackage{macros}

\usepackage{thm-restate}

\title{Expressive Power of Broadcast Consensus Protocols}

\titlerunning{Expressive Power of Broadcast Consensus Protocols}

\author{Michael Blondin}{Département d'informatique, Université de
  Sherbrooke, Sherbrooke, Canada}{michael.blondin@usherbrooke.ca}{}{
  Supported by the Fonds de recherche du Québec – Nature et
  technologies (FRQNT), by a Quebec--Bavaria project funded by the
  Fonds de recherche du Québec (FRQ), and by the Natural Sciences and
  Engineering Research Council of Canada (NSERC)}

\author{Javier Esparza}{ Fakultät für Informatik, Technische
  Universität München, Garching bei München,
  Germany}{esparza@in.tum.de}{}{Supported by an ERC Advanced Grant
  (787367: PaVeS)}

\author{Stefan Jaax}{ Fakultät für Informatik, Technische Universität
  München, Garching bei München, Germany} {jaax@in.tum.de}{}{Supported
  by an ERC Advanced Grant (787367: PaVeS)}

\authorrunning{M. Blondin, J. Esparza and S. Jaax}

\Copyright{John Q. Public and Joan R. Public}



\acknowledgements{Part of this work was realized while Stefan Jaax was
  visiting the Université de Sherbrooke.
  We warmly thank the anonymous reviewers for their helpful comments and suggestions.}

\nolinenumbers

\EventEditors{John Q. Open and Joan R. Access}
\EventNoEds{2}
\EventLongTitle{42nd Conference on Very Important Topics (CVIT 2016)}
\EventShortTitle{CVIT 2016}
\EventAcronym{CVIT}
\EventYear{2016}
\EventDate{December 24--27, 2016}
\EventLocation{Little Whinging, United Kingdom}
\EventLogo{}
\SeriesVolume{42}
\ArticleNo{23}
\begin{document}

\maketitle

%% Abstract
\begin{abstract}
  Population protocols are a formal model of computation by identical,
  anonymous mobile agents interacting in pairs.  Their computational
  power is rather limited: Angluin \textit{et al}. have shown that
  they can only compute the predicates over $\N^k$ expressible in
  Presburger arithmetic. For this reason, several extensions of the
  model have been proposed, including the addition of devices called
  cover-time services, absence detectors, and clocks. All these
  extensions increase the expressive power to the class of predicates
  over $\N^k$ lying in the complexity class $\NL$ when the input is
  given in unary. However, these devices are difficult to implement,
  since they require that an agent atomically receives messages from
  \emph{all} other agents in a population of unknown size; moreover, the
  agent must \emph{know} that they have all been received. Inspired by
  the work of the verification community on Emerson and Namjoshi's
  broadcast protocols, we show that $\NL$-power is also achieved by
  extending population protocols with reliable broadcasts, a simpler,
  standard communication primitive.
\end{abstract}

%% Contents
\section{Introduction}\label{sec:introduction}
Population protocols are a theoretical model for the study of ad hoc
networks of tiny computing devices without any
infrastructure~\cite{AADFP04,AADFP06}, intensely investigated in recent years (see
e.g.~\cite{AlistarhAEGR17,AlistarhAG18,AlistarhG18,ElsasserR18}). The
model postulates a ``soup'' of indistinguishable agents that behave
identically, and only have a fixed number of bits of memory, i.e., a
finite number of local states. Agents repeatedly interact in pairs,
changing their states according to a joint transition function. A
global fairness condition ensures that every finite sequence of
interactions that becomes enabled infinitely often is also executed
infinitely often. The purpose of a population protocol is to allow 
agents to collectively compute some information about their initial configuration, 
defined as the function that assigns to each local state the number of agents
that initially occupy it. For example, assume that initially each agent picks a
boolean value by choosing, say, $q_0$ or $q_1$ as its initial state. 
The many \emph{majority protocols} described in the
literature allow the agents to eventually reach a stable consensus on
the value chosen by a majority of the agents. More formally, let $x_0$
and $x_1$ denote the initial numbers of agents in states $q_0$ and $q_1$;
majority protocols compute the predicate $\varphi \colon \Nat \times \Nat \to \{0, 1\}$ given by $\varphi(x_0, x_1)
= (x_1 \geq x_0)$. Throughout the paper, we use the term ``predicate'' as an abbreviation
for ``function from $\N^k$ to $\{0,1\}$ for some $k$''.

In a seminal paper, Angluin~\emph{et al.}\ proved that population
protocols compute exactly the predicates expressible in Presburger
arithmetic~\cite{AADFP06,AAER07}. Thus, for example, agents can
decide if they are at least a certain number, if at least $2/3$ of them voted the same way, or, more generally, if the vector $(x_1,
x_2, \ldots, x_n)$ representing the number of agents that picked
option $1, 2, \ldots, n$ in an election with $n$ choices is a solution
of a system of linear inequalities. On the other hand, they
cannot decide if they are a square or a prime number, or if the
product of the number of votes for options $1$ and $2$ exceeds the
number of votes for option $3$. Much work has been devoted to designing more
powerful formalisms and analyzing their expressive power. In particular,
population protocols have recently been extended with capabilities allowing an agent to obtain
global information about the current configuration, which we proceed to describe. 

In \cite{MichailS15}, Michail 
and Spirakis extend the population protocol model
with \emph{absence detectors}, by means of which an agent knows, for every
state, whether the state is currently populated or not. Further, they implement absence detectors
by a weaker object called a \emph{cover-time service}, which allows an agent
to deduce if it has interacted with every other agent in the system. They prove that
protocols with cover-time can compute all predicates in $\DSPACE(\log n)$ and
can only compute predicates in $\NSPACE(\log n)=\NL$, where $n$ is the number of agents\footnote{Observe that, for example, $n$ agents can decide whether $n$ is prime. Indeed, a Turing machine can decide if $n$ is a prime number in $\Theta(\log n)$ space by going through all numbers from $2$ to $n-1$, and checking for each of them if they divide $n$.}.

In \cite{Aspnes17},  Aspnes observes that cover-time services are a kind of internal clock
mechanism, and introduces clocked population protocols. Clocked protocols have a clock
oracle that signals to one or more agents that the population has reached a bottom strongly connected component of the configuration graph, again an item of global information. Aspnes shows that clocked protocols can compute exactly the predicates in \NL. 

Absence detectors, cover-time services, and clocked protocols are difficult to implement, since they require that an agent reliably receives information from \emph{all} other agents; moreover, the agent needs to \emph{know} that it has already received messages from all other agents before making a move, which is particularly difficult because agents are assumed to have no identities and to ignore the size of the population. In this paper, we propose a much simpler extension (from an implementation point of view): We allow agents  to perform reliable broadcasts, a standard operation in concurrency and distributed computing. We are inspired by the broadcast protocol model introduced by Emerson and Namjoshi in~\cite{EmersonN98} to describe bus-based hardware protocols. The model has been used and further studied in many other contributions, e.g.~\cite{EsparzaFM99,FinkelL02,DelzannoRB02,SchmitzS13,BDGG17}. In
broadcast protocols, agents can perform binary interactions, as in the
population protocol model, but, additionally, an agent can also broadcast a signal to
all other agents, which are guaranteed to react to it. Broadcast protocols are rather simple to
implement with current technology on mobile agents moving in a limited area.  Broadcasts also
appear in biological systems. For example, Uhlendorf \textit{et al}. describe a system in which a controller adds a sugar or saline solution to a population of
yeasts, to which all the yeasts react~\cite{uhlendorf2015silico}. An idealized model of the system, which is essentially a broadcast protocol, has been analyzed by Bertrand \textit{et al}. in~\cite{BDGG17}.

In this paper, we show that population protocols with reliable broadcasts also compute \emph{precisely} the predicates in \NL, and are therefore as powerful as absence detectors or clocks. To prove this result, we first define the notion of \emph{silent semi-computation}, a weaker notion than standard computation, and prove that broadcast protocols silently semi-compute all protocols in \NL. This result makes crucial use of the ability of broadcast protocols to ``restart'' the whole population nondeterministically whenever something bad or unexpected is detected. We then prove that silent semi-computability and computability coincide for the class \NL.

In a second contribution, we explore in more detail the minimal requirements for achieving \NL\ power. On the one hand, we show that it is enough to allow \emph{a single} agent to broadcast \emph{a single} signal. On the other hand, we prove that the addition of a reset, which causes all agents to return to their initial states, does not increase the power of population protocols.

\section{Preliminaries}\label{sec:preliminaries}
\parag{Multisets} A \emph{multiset} over a finite set $E$ is a mapping $M \colon E \to
\N$. The set of all multisets over $E$ is denoted $\N^E$. For every $e
\in E$, $M(e)$ denotes the number of occurrences of $e$ in $M$. We
sometimes denote multisets using a set-like notation, \eg
$\multiset{f, g, g}$ is the multiset $M$ such that $M(f) = 1$, $M(g) =
2$ and $M(e) = 0$ for every $e \in E \setminus \{f, g\}$.  Addition
and comparison are extended to multisets componentwise, \ie $(M \mplus
M')(e) \defeq M(e) + M'(e)$ for every $e \in E$, and $M \leq M'
\defiff M(e) \leq M'(e)$ for every $e \in E$. We define multiset
difference as $(M \mminus M')(e) \defeq \max(M(e) - M'(e), 0)$ for
every $e \in E$. The empty multiset is denoted $\vec{0}$ and, for
every $e \in E$, we write $\vec{e} \defeq \multiset{e}$. Finally, we
define the \emph{support} and \emph{size} of $M \in \N^E$ respectively
as $\supp{M} \defeq \{e \in E : M(e) > 0\}$ and
$|M| \defeq \sum_{e \in E} M(e)$.

\parag{Population protocols}
A \emph{population} over a finite set $E$ is a multiset $P \in
\N^E$ such that $|P| \geq 2$. The set of all populations over $E$ is
denoted by $\pop{E}$. A \emph{population protocol with
  leaders} (population protocol for short) is a tuple $\PP = (Q, R, \Sigma, L, I, O)$ where:
\begin{itemize}
\item $Q$ is a non-empty finite set of \emph{states},

\item $R \subseteq (Q \times Q) \times (Q \times Q)$ is a set of \emph{rendez-vous
   transitions},

\item $\Sigma$ is a non-empty finite \emph{input alphabet},

\item $I \colon \Sigma \to Q$ is the \emph{input function} mapping
  input symbols to states,

\item $L \in \N^Q$ is the multiset of \emph{leaders}, and

\item $O \colon Q \to \{0, 1\}$ is the \emph{output function} mapping
  states to boolean values.
\end{itemize}

Following the standard convention, we call elements of
$\pop{Q}$ \emph{configurations}. Intuitively, a configuration $C$
describes a collection of identical finite-state \emph{agents} with
$Q$ as set of states, containing $C(q)$ agents in state $q$ for every
$q \in Q$, and at least two agents in total.

We write $(p, q) \mapsto (p', q')$ to denote that $(p, q, p', q') \in
 R$. The relation ${\Step} \colon \pop{Q} \to \pop{Q}$ is defined by:
 $(C, C') \in {\Step}$ if{}f there exists $(p, q, p', q') \in R$ such
 that $C \geq \multiset{p, q}$ and $C' = C \mminus \multiset{p,
 q} \mplus \multiset{p', q'}$. We write $C \trans{} C'$ if $(C,
 C') \in \Step$, and $C \trans{*} C'$ if $(C, C') \in \Step^*$, the
 reflexive and transitive closure of $\Step$. If $C \trans{*} C'$,
 then we say that $C'$ is
\emph{reachable} from $C$. An \emph{execution} is an infinite sequence
of configurations $C_0 C_1 \cdots$ such that $C_i \trans{} C_{i+1}$
for every $i \in \N$. An execution $C_0 C_1 \cdots$ is \emph{fair} if
for every step $C \trans{} C'$ the following holds: if $C_i = C$ for
infinitely many indices $i \in \N$, then $C_j = C'$
for infinitely many indices $j \in \N$.

We now explain the roles of the input function $I$ and the multiset
$L$ of leaders. The elements of $\pop{\Sigma}$ are called
\emph{inputs}. For every input $X \in \pop{\Sigma}$, let $I(X) \in
\pop{Q}$ denote the configuration defined by
\begin{align*}
  I(X)(q) &\defeq \sum_{\{\sigma \in \Sigma : I(\sigma) = q \}}
  X(\sigma) && \text{for every}\ q \in Q.
\end{align*}

A configuration $C$ is \emph{initial} if $C = I(X)+L$ for some input
$X$.  Intuitively, the agents of $I(X)$ encode the input, while those
of $L$ are a fixed number of agents, traditionally called leaders, that
perform the computation together with the agents of $I(X)$.

\medskip

\parag{Predicate computed by a protocol} If $O(p) = O(q)$ for every
$p, q \in \supp{C}$, then $C$ is a \emph{consensus configuration}, and
$O(C)$ denotes the unique output of the states in $\supp{C}$. We say
that a consensus configuration $C$ is a \emph{$b$-consensus} if $O(C)
= b$. An execution $C_0 C_1 \cdots$ \emph{stabilizes} to $b \in \{0,
1\}$ if there exists $n \in \N$ such that $C_i$ is a $b$-consensus for
every $i \geq n$.

A protocol $\PP$ over an input alphabet $\Sigma$
\emph{computes} a predicate $\varphi \colon \pop{\Sigma} \to \{0, 1\}$ if
for every input $X \in \pop{\Sigma}$, every fair execution of $\PP$
starting at the initial configuration $I(X)+L$ stabilizes to $\varphi(X)$.

Throughout the paper, we assume $\Sigma = \{A_1, \ldots, A_k\}$ for
some $k > 0$. Abusing language, we identify population $M \in
\pop{\Sigma}$ to vector $\vec{\alpha} = (M(A_1), \ldots, M(A_k))$, and
say that $\PP$ computes a \emph{predicate $\varphi \colon \Nat
  ^k \to \{0, 1\}$ of arity $k$}. In the rest of the paper, the term ``predicate''
  is used with the meaning ``function from $\Nat^k$ to $\{0,1\}$''. It is known that:

\begin{theorem}[{\cite{AAER07}}]
Population protocols compute exactly the predicates expressible in Presburger arithmetic, i.e. the first-order theory of the natural numbers with addition.
\end{theorem}

\section{Broadcast consensus protocols}
Broadcast protocols were introduced by Emerson and Namjoshi
in~\cite{EmersonN98} as a formal model of bus-based hardware protocols, such as
those for cache coherency. The model has also been applied to the
verification of multithreaded programs~\cite{DelzannoRB02}, and to
idealized modeling of control problems for living
organisms~\cite{uhlendorf2015silico,BDGG17}. Its theory has
been further studied in~\cite{EsparzaFM99,FinkelL02,SchmitzS13}.

Agents of broadcast protocols can communicate in pairs, as in
population protocols, and, additionally, they can also communicate by
means of a reliable broadcast. An agent can broadcast a signal to all
other agents, which after receiving the signal move to a new state.
Broadcasts are routinely used in wireless ad-hoc and sensor networks
(see e.g.~\cite{AbolhasanWD04,YickMG08}), and so they are easy to
implement on the same kind of systems targeted by population
protocols. They can also model idealized versions of communication in
natural computing. For example, in~\cite{BDGG17} they are used to
model ``communication'' in which an experimenter ``broadcasts'' a signal to
a colony of yeasts by increasing the concentration of a nutrient in a
solution.

We introduce broadcast consensus protocols, i.e., broadcast protocols
whose goal is to compute a predicate in the computation-by-consensus
paradigm.

\begin{definition}\label{def:broadcast}
A broadcast consensus protocol is a tuple $\PP = (Q, R, B, \Sigma, L, I, O)$,
where all components but $B$ are defined as for population protocols, and $B$ is a set
of \emph{broadcast transitions}. A \emph{broadcast transition} is a triple $(q,r,f)$
where $q,r \in Q$ and $f \colon Q \rightarrow Q$ is a \emph{transfer function}.

The relation ${\Step} \subseteq \pop{Q} \times \pop{Q}$ of $\PP$ is defined as follows. A pair $(C, C')$ of configurations belongs to ${\Step}$ if{}f
\begin{itemize}
\item there exists $(p, q) \mapsto (p', q') \in R$ such that $C \geq \multiset{p, q}$ and $C' = C
 \mminus \multiset{p, q} \mplus \multiset{p', q'}$; or
\item there exists a transition $(q, r, f) \in B$ such that $C(q) \geq 1$ and $C'$
 is the configuration computed from $C$ in the following three steps:
 \begin{align}
   C_1     &= C \mminus \multiset{q}, \\
   C_2(q') &= \sum_{r' \in f^{-1}(q')} C_1(r') & \text{for every}~q' \in Q,\\
   C'      &= C_2 \mplus \multiset{r}.
 \end{align}
 \end{itemize}
\end{definition}

Intuitively, (1)--(3) is interpreted as follows: (1)~an agent at
state $q$ broadcasts a signal and leaves $q$, yielding $C_1$; (2)~all other agents receive the signal
and move to the states indicated by the function $f$, yielding $C_2$; and (3)~the
broadcasting agent enters state $r$, yielding $C'$.  Correspondingly, instead of $(q,r,f)$
we use ${q \mapsto r}; \, f$ as notation for a broadcast transition.\medskip

\newcommand{\transfer}{\mapsto}

\parag{Beyond Presburger arithmetic} As a first illustration of the power of broadcast protocols, we show
that their expressive power goes beyond Presburger arithmetic, and so
beyond the power of population protocols. We present a
broadcast consensus protocol for the predicate $\varphi$, defined as $\varphi(x)
= 1$ if{}f $x > 1$ and $x$ is a power of two. For readability, we
use the notation $q \mapsto q';
[q_1 \transfer q_1', \ldots, q_n \transfer q_n']$ for a broadcast transition, where $f(q_i) =
q_i'$ and where transfers of the form $q_i \transfer q_i$ may be omitted.

Let $\PP = (Q, R, B, \Sigma, L, I, O)$ be the broadcast consensus
protocol where $Q \defeq \{x, \overline{x}, \tilde{x}, 0,
1,\allowbreak \bot\}$, $\Sigma \defeq \{x\}$, $I \defeq x \mapsto x$,
$L \defeq \vec{0}$, $O(q) = 1 \defiff q = 1$, and $R$ and $B$ are
defined as follows:
\begin{itemize}
\item $R$ contains the rendez-vous transition $s \colon (x, x) \mapsto (\overline{x}, 0)$;

\newcommand{\maketransfer}[2]{\mathmakebox[6pt][r]{#1} \transfer \mathmakebox[6pt][l]{#2}}

\item $B$ contains the broadcast transitions $r \colon \bot \mapsto x
  ; \left[q \mapsto x : q \in Q\right]$ and
  {\small $$\begin{array}{rlllllll}
  %% ¯s
  \overline{s} \colon
  \overline{x} \mapsto x;
  \left[\!\!\begin{array}{l}
      \maketransfer{x}{\bot} \\
      \maketransfer{\overline{x}}{x} \\
      \maketransfer{0}{1}
    \end{array} \!\!\right]
  &
  %% t0¯
  \overline{t_0} \colon
  \overline{x} \mapsto \overline{x};
  \left[\!\!\begin{array}{l}
      \maketransfer{1}{0}
    \end{array} \!\!\right]
  &
  %% t0
  t_0 \colon
  x \mapsto 0;
  \left[\!\!\begin{array}{l}
      \maketransfer{x}{\bot} \\
      \maketransfer{\overline{x}}{0} \\
      \maketransfer{1}{\bot}
    \end{array} \!\!\right]
  &
  %% t1
  t_1 \colon
  x \mapsto 1;
  \left[\!\!\begin{array}{l}
      \maketransfer{x}{\bot} \\
      \maketransfer{\overline{x}}{\bot} \\
      \maketransfer{0}{\bot}
    \end{array} \!\!\right].
\end{array}$$
}
  \end{itemize}
  Intuitively, $\PP$ repeatedly halves the number of agents in state
  $x$, and it accepts if{f} it never obtains an odd remainder. More
  precisely, the transitions of $\PP$ are intended to be fired as
  follows, where $C$ denotes the current configuration:
  \begin{quote}
    \newcommand{\inst}[1]{\textbf{#1}}
    \newcommand{\icol}{:\;}
    \newcommand{\ind}{\quad}
    \newcommand{\comm}[1]{\hspace{1cm}\textcolor{black!60!gray}{\small/*~\makebox[6.5cm][r]{#1}~*/}}
    \ttfamily
    \begin{tabular}{lr}
      \inst{while} $C(x) \neq 1$\icol \\
      \ind \inst{while} $C(x) \geq 2$\icol \inst{fire} $s$
      & \comm{split agents equally from $x$ to $\overline{x}$ and $0$} \\
      \ind \inst{if} $C(x) = 0$\icol \inst{fire} $\overline{s}$
      & \comm{move agents from $\overline{x}$ to $x$ if no remainder} \\[1.5pt]
      \inst{if} $C(\overline{x}) = 0$\icol \inst{fire} $t_1$
      & \comm{if no remainder, then accept} \\
      \inst{else}\icol \inst{fire} $\overline{t_0}\, t_0$
      & \comm{otherwise, reject}
    \end{tabular}
  \end{quote}
  It is easy to show that $\PP$ produces a (lasting) consensus, and
  the right one, if transitions are executed as above. However, an
  arbitrary execution may not follow the above procedure. Firing
  transition $\overline{t_0}$ when not intended has no incidence on
  the outcome. Moreover, if another transition is fired when it should
  not be, then $\overline{s}$, $t_0$ or $t_1$ will detect this error
  by moving an agent to state $\bot$. In this case, by fairness, $r$
  eventually resets the agents back to the initial configuration and,
  again by fairness, transitions are eventually fired as intended.

\begin{restatable}{proposition}{propNonSemilin}\label{prop:non:semilin}
  The broadcast consensus protocol $\PP$ described above computes the
  predicate $\varphi$, defined as $\varphi(x)= 1$ if{}f $x > 1$
  and $x$ is a power of two.
\end{restatable}

\parag{Leaderless broadcast protocols}
A broadcast protocol $\PP = (Q, R, B, \Sigma, L, I, O)$ is \emph{leaderless}
if $L = \vec{0}$.
It can be shown that leaderless
broadcast consensus protocols compute the same predicates as the general class.
We only sketch the argument. First, a broadcast protocol with
leader multiset $L$ can be simulated by a protocol with a single
leader. Indeed, the protocol can be designed so that the first task of
the leader is to ``recruit'' the other leaders of $L$ from among the
agents. Second, a protocol with one leader can be simulated
by a leaderless protocol because, loosely speaking, a broadcast
protocol can elect a leader in a single computation step\footnote{Unlike population protocols, where efficient leader election is
non-trivial and much studied; see e.g.~\cite{ElsasserR18}.}. Indeed,
if initially all agents are in a state, say $q$, then
a broadcast $q \mapsto \ell; f$, where $f(q) = q'$, sends exactly one
agent to leader state $\ell$, and all other agents to state $q'$. It
is simple to construct $\PP'$ using this feature, and the details are
omitted.

In the rest of the paper, we use protocols with leaders to simplify the constructions, but all results
(except Proposition~\ref{prop_single_signal}) remain valid for leaderless protocols.

\section{Broadcast consensus protocols compute exactly $\NL$}
In this section, we prove our main theorem: a predicate is computable
by a broadcast consensus protocol if{}f it is in $\NL$. We follow the
convention and say that a predicate $\varphi$ belongs to $\NL$ if
there is a nondeterministic Turing machine that accepts in $\bigO(\log
n)$-space exactly the tuples $(x_1, x_2, \ldots, x_k) \in \N^k$,
encoded in unary, such that $\varphi(x_1, x_2, \ldots, x_k)$ holds.

The proof is divided in two parts. Section~\ref{sec:easy} proves the
easier direction: predicates computable by broadcast consensus
protocols are in $\NL$. Section~\ref{sec:hard} proves the converse,
which is more involved.

\subsection{Predicates computable by broadcast consensus protocols are in $\NL$}
\label{sec:easy}

We prove the result in more generality. We define a generic computational
model in which the possible steps between configurations are given
by an arbitrary relation preserving the number of agents. Formally,
a \emph{generic consensus protocol} is a tuple $\PP = (Q, \Step, \Sigma, L, I, O)$ where
$Q, \Sigma, L, I, O$ are defined as for population protocols, and
$\Step \subseteq \pop{Q} \times \pop{Q}$ is the \emph{step
relation} between populations, satisfying $|C| = |C'|$ for every
$(C, C') \in \Step$.

Clearly, broadcast consensus protocols are generic consensus protocols. Further,
it is easy to see that if $\Step$ is the one-step relation of a
broadcast protocol, then $\Step \in \NL$. Indeed, $\Step \in \NL$ if
there is a nondeterministic Turing machine that given a pair of
configurations $(C, C')$ with $n$ agents, uses $\bigO(\log n)$ space and
accepts if{}f $(C, C') \in \Step$.  A quick inspection of the two
conditions in the definition of $\Step$
(Definition \ref{def:broadcast}) shows that this is the case.

Thus, it suffices to prove that generic consensus protocols satisfying $\Step \in \NL$ can only compute predicates in $\NL$. We sketch the proof, more details can be found in the Appendix.

\begin{restatable}{proposition}{propGenUpperBound}\label{prop:gen:upper:bound}
  Let $\PP = (Q, \Step, \Sigma, L, I, O)$ be a generic consensus protocol
  computing a predicate $\varphi$. If $\Step \in \NL$,
  then $\varphi \in \NL$. In particular, predicates computable by broadcast consensus protocols are in $\NL$.
\end{restatable}
\begin{proof}
  We show that there is a nondeterministic Turing machine that decides
  whether $\varphi(\vec{x}) = 1$ holds, and uses $\bigO(\log
  |\vec{x}|)$ space. Let $G = (V, E)$ be the graph where $V$ is the
  set of all configurations of $\PP$ of size $|\vec{x}|$, and $(C, C')
  \in E$ if{}f $C \trans{} C'$.

  It is easy to see that $\varphi(\vec{x}) = 1$ if{}f $G$ contains a
  configuration $C$ of size $|C|=|I(\vec{x})| = |\vec{x}|$
  satisfying (1) $C_0 \trans{*} C$;\label{prop:reach} and (2) every
  configuration reachable from $C$, including $C$ itself, is a
  $1$-consensus. Therefore, we can decide $\varphi(\vec{x}) = 1$ by
  guessing $C$, and checking~(1) and~(2) in $\bigO(\log
  |I(\vec{x})|)$ space.  For~(1), this follows from the fact that
  graph reachability is in \NL. For (2), we observe that determining
  whether some configuration reachable from $C$ is not a $1$-consensus
  can be done in \NL, and we use the fact that $\NL =
  \coNL$~\cite{Imm88}.
  \end{proof}

\begin{remark}
Protocols with absence detector~\cite{MichailS15} are a class of generic consensus protocols, and hence Proposition~\ref{prop:gen:upper:bound} can be used to give an alternative proof of the fact that these protocols only compute predicates in $\NL$.
\end{remark}

\subsection{Predicates in $\NL$ are computable by broadcast consensus protocols}
\label{sec:hard}

The proof is involved, and we start by describing its structure. In
Section~\ref{sub:semi}, we show that it suffices to prove that every
predicate in $\NL$ is \emph{silently
semi-computable}. In the rest of the section, we proceed to prove this
in three steps. Loosely speaking, we show that:
\begin{itemize}
\item predicates computable by nondeterministic Turing machines in
  $\bigO(n)$ space can also be computed by counter machines with
  counters polynomially bounded in $n$ (Section
  \ref{subsec:simu1});

\item predicates computed by polynomially bounded counter machines can
  also be computed by $n$-bounded counter machines, i.e.\ in which the
  sum of the values of all counters never exceeds their initial sum
  (Section~\ref{subsec:simu2});

\item predicates computed by $n$-bounded counter machines can be
  silently semi-computed by broadcast protocols.
  (Section~\ref{subsec:simu3}).
\end{itemize}
\noindent Finally, Section~\ref{subsec:main} puts all parts of the proof
together.

\subsubsection{Silent semi-computation}\label{sub:semi}

Recall that, loosely speaking, a protocol computes $\varphi$ if it
converges to $1$ for inputs that satisfy $\varphi$, and it converges
to $0$ for inputs that do not satisfy $\varphi$. Additionally, a
protocol \emph{silently computes} $\varphi$ if convergence to $b \in
\{0,1\}$ happens by reaching a \emph{terminal $b$-consensus}, i.e., a
configuration $C$ that is a $b$-consensus and from which one can only
reach $C$ itself. (Intuitively, the protocol eventually becomes
``silent'' because no agent changes state anymore, and hence
communication ``stops''.)  We say that a protocol \emph{silently
  semi-computes} $\varphi$ if it reaches a terminal $1$-consensus for
inputs that satisfy $\varphi$, and no terminal configuration for other
inputs.

\begin{definition}\label{definition:silently-semicompute}
  A broadcast consensus protocol $\PP$ \emph{silently semi-computes} a $k$-ary
  predicate $\varphi$ if for every $\vec{\alpha} \in \N^k$ the
  following properties hold:
  \begin{enumerate}
  \item if $\varphi(\vec{\alpha}) = 1$, then every fair execution of
    $\PP$ starting at $I(\vec{\alpha})$ eventually reaches a terminal
    $1$-consensus configuration;\label{itm:semi:1}

  \item if $\varphi(\vec{\alpha}) = 0$, then no fair execution of $\PP$
  starting at $I(\vec{\alpha})$ eventually reaches a terminal configuration.\footnote{Since every finite execution can be extended to a fair one, this condition is actually equivalent to ``no terminal configuration is reachable from $I(\vec{\alpha})$''.}\label{itm:semi:2}
  \end{enumerate}
\end{definition}

We show that if a predicate and its complement are
both silently semi-computable by broadcast consensus protocols, say $\PP_1$ and
$\PP_0$, then the predicate is also computable by a broadcast consensus protocol
$\PP$ which, intuitively,
behaves as follows under input $\vec{\alpha}$. At every moment in time, $\PP$ is simulating
either $\PP_1$ or $\PP_0$. Initially, $\PP$ simulates $\PP_0$. Assume
$\PP$ is simulating $\PP_i$ and the current configuration is $C$. If
$C$ is a terminal configuration of $\PP_i$, then $\PP$ terminates
too. Otherwise, $\PP$ nondeterministically chooses one of three options: continue thesimulation of $\PP_i$, ``reset'' the computation to $I_{0}(\vec{\alpha})$,
i.e., start simulating $\PP_0$, or ``reset'' the computation to $I_{1}(\vec{\alpha})$. Conditions~\ref{itm:semi:1}
and~\ref{itm:semi:2} ensure that exactly one of $\PP_0$ and $\PP_1$
can reach a terminal configuration, namely
$\PP_{\varphi(\vec{\alpha})}$. Fairness ensures that $\PP$ will
eventually reach a terminal configuration of
$\PP_{\varphi(\vec{\alpha})}$, and so, by condition~\ref{itm:semi:1},
that it will always reach the right consensus. Hence, $\PP$ silently computes
$\varphi$.

The ``reset'' is implemented by means of a broadcast that sends every
agent to its initial state in the configuration
$I_j(\vec{\alpha})$; for this, the states of $\PP$ are partitioned
into classes, one for each input symbol $x \in X$. Every agent moves
only within the states of one of the classes, and so every agent
``remembers'' its initial state in both $\PP_0$ and $\PP_1$.

\begin{lemma}\label{lemma:semi-to-nonsemi}
  Let $\varphi$ be an $m$-ary predicate, and let $\overline{\varphi}$
  be the predicate defined by $\overline{\varphi}(\vec{\alpha}) \defeq
  1 - \varphi(\vec{\alpha})$ for every $\vec{\alpha} \in
  \Nat^m$. Further let $\PP_1$ and $\PP_0$ be broadcast consensus
  protocols that silently semi-compute $\varphi$ and
  $\overline{\varphi}$, respectively. The following holds: there
  exists a broadcast consensus protocol $\PP$ that silently computes
  $\varphi$.
\end{lemma}

\begin{proof}
  Let $\PP_1 = (Q_1, R_1, B_1, \Sigma, I_1, O_1)$ and $\PP_0 =
  (Q_0, R_0, B_0, \Sigma, I_0, O_0)$ be protocols that silently
  semi-compute $\varphi$ and $\overline{\varphi}$,
  respectively. Assume w.l.o.g.\ that $Q_1$ and $Q_0$ are disjoint. We
  construct a protocol $\PP = (Q, R, B, \Sigma, \allowbreak I, O)$
  that computes $\varphi$.

For the sake of clarity we refrain from giving a fully formal
description, but we provide enough details to show that the design
idea above can indeed be implemented.

\parag{States and mappings} The set of states of $\PP$ is defined as:
$$Q \defeq \Sigma \times
  (Q_1 \cup Q_0 \cup \{ \reset \})$$
  If an agent is in state $(x, q)$, we say that
$x$ is its \emph{origin} and that $q$ is its \emph{position}. The
initial position of an agent is its initial state in $\PP_0$, i.e.
$I(x) \defeq (x, I_0(x))$. Transitions will be designed so that
agents may update their position, but not their origin.
Alternatively, instead of applying a transition, agents can nondeterministically choose to
transition from $(x, q) \in X \times (Q_1 \cup Q_0)$ to $(x, \reset)$.
An agent in state $(x, \reset)$ eventually resets the simulation to either $\PP_0$ or $\PP_1$.

\parag{Simulation transitions} We define transitions that proceed with
the simulation of $\PP_0$ and $\PP_1$ as follows. For every $i \in
\{1, 0\}$, every $x, y \in \Sigma$, and every rendez-vous transition $(q,
r) \mapsto (q', r')$ of $R_i$, we add the following
rendez-vous transitions to $R$:
$$(x, q), (y, r)  \mapsto (x, q'), (y, r') \qquad \mbox{ and } \qquad (x, q), (y, r)  \mapsto (x, q'), (y, r')  .$$
The first transition implements the simulation, while the second transition enables resets when the simulation has not reached a terminal configuration.
For every broadcast transition $q \mapsto q'; f$ of $B_i$ and
every $x \in \Sigma$, we add the following broadcast transitions to $B$:
\begin{align*}
(x, q) & \mapsto (x, q'); f' \\
(x, q) & \mapsto (x, \reset); f'
\end{align*}
where $f'$ only acts on $Q_i$ by $f'(y, r) \defeq (y, f(r))$ for every
$(y, r) \in \Sigma \times Q_i$.
The first transition implements the simulation of a broadcast in the original protocols, while the second transition enables a reset.

\parag{Reset transitions} We define transitions that trigger a new
simulation of either $\PP_0$ or $\PP_1$. For every $i \in \{1, 0\}$, let $f_i \colon
Q \to Q$ be the function defined as $f_i(x, q) \defeq (x, I_i(x))$ for
every $(x, q) \in Q$. For every $i \in \{1, 0\}$ and every $x \in \Sigma$, we add the following broadcast
transition to $B$:
$(x, \reset) \mapsto (x, I_i(x)); f_i$.
\end{proof}

Using Lemma \ref{lemma:semi-to-nonsemi}, we may now prove the following:

\begin{proposition}\label{prop:fromsemitocomp}
  If every predicate in $\NL$ is silently
  semi-computable by broadcast consensus protocols, then every predicate
  in $\NL$ is silently computable (and so computable) by broadcast consensus protocols.
\end{proposition}

\begin{proof}
  Assume every predicate in $\NL$ is silently
  semi-computable by broadcast consensus protocols, and let $\varphi$ be a
  predicate in $\NL$. We resort to the powerful result
  stating that predicates in $\coNL$ and $\NL$
  coincide. This is an immediate corollary of the $\coNL
  = \NL$ theorem for
  languages~\cite{Imm88,Szelepcsenyi88,Papadimitriou07}, and the fact
  that one can check in constant space whether a given word encodes a vector of natural numbers of fixed arity. Thus, both $\varphi$ and
  $\overline{\varphi}$ are predicates in $\NL$, and so,
  by assumption, silently semi-computable by broadcast consensus protocols. By
  Lemma~\ref{lemma:semi-to-nonsemi}, they are silently computable by broadcast
  consensus protocols.
\end{proof}

\subsubsection{Simulation of Turing machines by counter machines}\label{subsec:simu1}

We recall that nondeterministic Turing machines working in $\bigO(n)$
space can be simulated by counter machines whose counters are
polynomially bounded in $n$, and so that both models compute the same
predicates.

Let $X = \{x_1, x_2, \ldots, x_k\}$ and $\textit{Ins} = \{\incr(x),
\decr(x), \zero(x), \nonzero(x), \nop \mid x \in X\}$.  A
\emph{$k$-counter machine} $\M$ over \emph{counters} $X$ is a tuple
$(Q, X, \Delta, m, q_0, q_a, q_r)$, where $Q$ is a finite set of
\emph{control states}; $\Delta \subseteq Q \times \textit{Ins} \times
Q$ is the \emph{transition relation}; $m \leq k$ is the \emph{number
  of input counters}; and $q_0, q_a, q_r$ are the \emph{initial},
\emph{accepting}, and \emph{rejecting states}, respectively.

A configuration of $\M$ is a pair $C = (q, \vec{v}) \in Q \times \N^k$
consisting of a control state $q$ and counter values $\vec{v}$. For
every $i \in [k]$, we denote the value of counter $x_i$ in $C$ by
$C(x_i) \defeq \vec{v}_i$. The \emph{size} of $C$ is $|C| \defeq
\sum_{i=1}^k C(x_i)$.

\newcommand{\etrans}[1]{\trans{\mathmakebox[17pt]{#1}}}
\newcommand{\ins}{\textit{ins}}
Let $\vec{e}_i$ be the $i$-th row of the $k \times k$ identity
matrix.  Given $\ins \in \textit{Ins}$, we define
the relation $\trans{\ins}$ over configurations as follows: $(q, \vec{v}) \trans{\ins} (q', \vec{v}')$
if{}f $(q, \ins, q')  \in \Delta$ and one of the following holds:  $\ins= \incr(x_i) $ and $\vec{v}' = \vec{v} + \vec{e}_i$;  $\ins = \decr(x_i)$, $\vec{v}_i > 0$, and $\vec{v}' =\vec{v} - \vec{e}_i$; $\ins = \zero(x_i)$, $\vec{v}_i = 0$, and $\vec{v}' = \vec{v}$; $\ins = \nonzero(x_i)$, $\vec{v}_i > 0$, and $\vec{v}' = \vec{v}$; $\ins = \nop$ and $\vec{v}' = \vec{v}$.

For every $\vec{\alpha} \in \N^m$, the initial configuration of $\M$
with input $\vec{\alpha}$ is defined as:
\begin{align*}
  C_{\vec{\alpha}} &\defeq (q_0, (\vec{\alpha}_1, \vec{\alpha}_2,
  \ldots, \vec{\alpha}_m, \underbrace{0, \ldots, 0}_{k -
    m~\text{times}})).
\end{align*}
We say $\M$ \emph{accepts} $\vec{\alpha}$ if there exist counter
values $\vec{v} \in \N^k$ satisfying $C_{\vec{\alpha}} \trans{*} (q_a,
\vec{v})$. We say $\M$ \emph{rejects} $\vec{\alpha}$ if $M$ does not
accept $\vec{\alpha}$ and for all configurations $C'$ with
$C_{\vec{\alpha}} \trans{*} C'$, there exists $\vec{v}
\in \N^k$ satisfying $C' \trans{*} (q_r, \vec{v})$. We say $\M$
\emph{computes} a predicate $\varphi \colon \N^m \to \{0, 1\}$
if $\M$ accepts all inputs $\vec{\alpha}$ such that
$\varphi(\vec{\alpha}) = 1$, and rejects all $\vec{\alpha}$ such that
$\varphi(\vec{\alpha}) = 0$.

A counter machine $\M$ is \emph{$f(n)$-bounded} if $|C| \leq
f(|C_{\vec{\alpha}}|)$ holds for every initial configuration
$C_{\vec{\alpha}}$ and every configuration $C$ reachable from
$C_{\vec{\alpha}}$. It is well-known that counter machines can
simulate Turing machines:

\begin{theorem}[{\cite[Theorem 3.1]{FMR68}}]\label{thm:TM-CM}\label{thm:FMR68}
  A predicate is computable by an $s(n)$-space-bounded Turing
  machine if{}f it is computable by a $2^{s(n)}$-bounded counter
  machine.
\end{theorem}

In~\cite{FMR68}, a weaker version of Theorem~\ref{thm:TM-CM} is proven
that applies to deterministic Turing and counter machines only. However,
the proof can be easily adapted to the nondeterministic setting we
consider here.

\begin{corollary}
  A predicate is in $\NL$ if{}f it is computable by a
  polynomially bounded counter machine.
\end{corollary}

\subsubsection{Simulation of polynomially bounded counter machines by $n$-bounded counter machines}\label{subsec:simu2}

\begin{restatable}{lemma}{lemPolyInput}\label{lem:poly:input}
  For every polynomially bounded counter machine that computes some
  predicate $\varphi$, there exists an $n$-bounded counter machine
  that computes $\varphi$.
\end{restatable}

\newcommand{\Mo}{\overline{\M}}
\begin{proof}
  We sketch the main idea of the proof; details can be found in the
  Appendix. Let $c \in \N_{> 0}$ and let $\M$ be an $n^c$-bounded
  counter machine with $k$ counters. To simulate $\M$ by an
  $n$-bounded counter machine $\Mo$, we need some way to represent any
  value $\ell \in [0, n^c]$ by means of counters with values in $[0,
  n]$. We encode such a value $\ell$ by its base $n + 1$
  representation over $c$ counters. Zero-tests are performed by
  zero-testing all $c$ counters sequentially. Nonzero-tests are
  implemented similarly with parallel tests. Incrementation and
  decrementation are implemented with gadgets to (a) assign $0$ to a
  counter; (b) assign $n$ to a counter; (c) test whether a counter
  value equals $n$.

  This construction is only \emph{weakly} $n$-bounded, in the sense
  that all counters are indeed bounded by $n$, but the overall sum can
  reach $k \cdot n$. To circumvent this issue, we simulate $\Mo$ by
  another counter machine $\M'$ whose counters symbolically hold
  values from multiple counters of $\Mo$. In more details, the
  counters are defined as $\{y_S : S \subseteq
  \overline{X}\}$. Intuitively, if counter $y_S$ has value $a$, then
  it contributes by $a$ to the value of each counter of $S$. For
  example, if $\overline{X} = \{x_1, x_2, x_3\}$ and the input size is
  $n = 6$, then counter values $(x_1, x_2, x_3) = (6, 1, 4)$ of $\Mo$
  can be represented in $\M'$ as $y_{\{x_1, x_2, x_3\}} = 1$,
  $y_{\{x_1, x_3\}} = 3$, $y_{\{x_1\}} = 2$, and $y_S = 0$ for every
  other $S$. Under such a representation, the sum of all counters
  equals $n$. Moreover, all instructions can be implemented quite
  easily.
\end{proof}

\subsubsection{Simulation of $n$-bounded counter machines by broadcast consensus protocols}\label{subsec:simu3}

Let $\M = (Q, X, \Delta, m, q_0, q_a, q_r)$ be an $n$-bounded counter
machine that computes some predicate $\varphi \colon \N^m \to \{0,
1\}$. We construct a broadcast consensus protocol $\PP= (Q', R, B,
\Sigma, L, I, O)$ that silently semi-computes $\varphi$.

\newcommand{\idle}{\textit{idle}}
\newcommand{\error}{\textit{err}}
\newcommand{\res}{\textit{rst}}
\newcommand{\one}{\textit{one}}
\newcommand{\Error}{\textit{Move-to-error}}
\renewcommand{\Reset}{\textit{Back-to-origin}}
\newcommand{\One}{\textit{Move-to-one}}
\newcommand{\send}{\rightharpoonup}
\newcommand{\rec}{\rightharpoondown}

\parag{States and mappings} Let $X' \defeq X \cup \{\idle,
\error\}$. The states of $\PP$ are defined as
$$Q' \defeq \underbrace{Q \times \{0, 1\}}_{\text{leader states}}\
\cup\ \underbrace{X' \times X \times \{0,
  1\}}_{\text{nonleader states}}.
$$ The protocol will be designed in such a way that there is always
exactly one agent, called the \emph{leader}, in states $Q \times \{0,
1\}$. Whenever the leader is in state $(q, b)$, we say that its
\emph{position} is $q$, and its \emph{opinion} is $b$. Every other
agent will remain in a state from $X' \times X \times \{0,
1\}$. Whenever a nonleader agent is in state $(x, y, b)$, we say that
its \emph{position} is $x$, its \emph{origin} is $y$, and its
\textit{opinion} is $b$. Intuitively, the leader is in charge of
storing the control state of $\M$, and the nonleaders are in charge of
storing the counter values of $\M$.

The protocol has a single leader whose initial position is the initial
control state of $\M$, i.e.\ $L \defeq \multiset{(q_0, 0)}$. Moreover,
every nonleader agent initially has its origin set to its initial
position, which will remain unchanged by definition of the forthcoming
transition relation: $I(x) \defeq (x, x, 0)~\text{for every}~x \in X$.
The output of each agent is its opinion:
$$
\begin{array}{rl}
     O(q, b) \defeq b &
  \multirow{2}{*}{for every $q \in Q, x \in X', y \in X, b \in \{0, 1\}$.} \\
  O(x, y, b) \defeq b &
\end{array}
$$

We now describe how $\PP$ simulates the instructions of $\M$.

\parag{Decrementation/incrementation} For every transition $q
\trans{\decr(x)} r \in \Delta$, every $y \in X$ and every $b, b' \in
\{0, 1\}$, we add to $R$ the rendez-vous transition:
$$(q, b), (x, y, b') \mapsto (r, b), (\idle, y, b').$$ These
transitions change the position of one agent from $x$ to $\idle$, and
thus decrement the number of agents in position $x$.

Similarly, for every transition $q \trans{\incr(x)} r$, every $y \in
X$ and every $b, b' \in \{0, 1\}$, we add to $R$ the rendez-vous
transition:
$$(q, b), (\idle, y, b') \mapsto (r, b), (x, y, b').$$ These
transitions change the position of an idle agent to $x$, and thus
increment the number of agents in position $x$. If no agent is in
position $\error$, then at least one idle agent is available when
a counter needs to be incremented, since $\M$ is $n$-bounded.

\parag{Nonzero-tests} For every $q \trans{\nonzero(x)} r
\in \Delta$, every $y \in X$ and every $b, b' \in \{0, 1\}$, we add to
$R$ the rendez-vous transition:
$$(q, b), (x, y, b') \mapsto (r, b), (x, y, b').$$ These transitions
can only be executed if there is at least one agent in position
$x$, and thus only if the value of $x$ is nonzero.

\parag{Zero-tests} For a given $x \in X$, let $f^x_\error \colon Q' \to Q'$ be the function
that maps every nonleader in position $x$ to the error position, i.e. $f^x_\error(x, y, b) \defeq (\error, y, b)$ for every $y \in X, b \in \{0, 1\}$, and $f^x$ is the identity for all other states.

For every transition $q \trans{\zero(x)} r \in \Delta$ and every $b
\in \{0, 1\}$, we add to $B$ the broadcast transition $(q, b) \mapsto (r, b); f^x_\error$. If such a transition occurs, then nonleaders in
position $x$ move to $\error$. Thus, an error is detected if{}f the
value of $x$ is nonzero.

To recover from errors, $\PP$ can be reset to its initial
configuration as follows. Let $f_\res \colon Q' \to Q'$ be the
function that sends every state back to its origin, i.e.\
\begin{align*}
  f_\res(q, b) &\defeq(q_0, 0)
  &&\text{for every}~q \in Q, b \in \{0, 1\}, \\
  f_\res(x, y, b) &\defeq (y, y, 0)
  &&\text{for every}~x \in X', y \in X, b \in \{0, 1\}.
\end{align*}
For every $y \in X$ and every $b \in \{0, 1\}$, we add the following
broadcast transition to $B$ to reset $\PP$ to its initial
configuration:
$$
(\error, y, b) \mapsto (y, y, 0); f_\res.
$$

\parag{Acceptance} For every $q \in Q \setminus \{q_a\}$ and $b \in
\{0, 1\}$, we add to $B$ the broadcast transition $(q, b) \mapsto (q_0, 0); f_\res$. Intuitively, as long as the leader's position differs from the
accepting control state $q_a$, it can reset $\PP$ to its initial
configuration. This ensures that $\PP$ can try \emph{all}
computations.

Let $f_\error \colon Q' \to Q'$ be the function that changes the
opinion of each state to $1$, i.e.\
\begin{align*}
  f_\error(q, b) &\defeq(q, 1)
  &&\text{for every}~q \in Q, b \in \{0, 1\}, \\
  f_\error(x, y, b) &\defeq (x, y, 1)
  &&\text{for every}~x \in X', y \in X, b \in \{0, 1\}.
\end{align*}
For every $b \in \{0, 1\}$, we add the following transition to $B$:
$$t_{\one, b} \colon (q_a, b) \to (q_a, 1); f_\one.$$ Intuitively,
these transitions change the opinion of every agent to $1$. If such a
transition occurs in a configuration with no agent in $\error$, then
no agent can change its state anymore, and the stable consensus $1$
has been reached.

\parag{Correctness} Let us fix some some input $\vec{\alpha} \in
\N^m$. Let $C_0$ and $D_0$ be respectively the initial configurations
of $\M$ and $\PP$ on input $\vec{\alpha}$. Abusing notation, for every
$D \in \pop{Q'}$, let
$$D(x) \defeq \sum_{(x, y, b) \in Q'} D(x, y, b).$$

The two following propositions state that every execution of $\M$ has
a corresponding execution in $\PP$ and vice versa. The proofs are routine.

\begin{proposition}\label{proposition:weak-simulation}
  Let $C$ be a configuration of $\M$ such that $C$ is in control state
  $q$ and $C_0 \trans{*} C$. There exists a configuration $D \in
  \pop{Q'}$ such that (i) $D_0 \trans{*} D$;  (ii) $D(x) = C(x)$ for every $x \in X$;
  (iii) $D(\error) = 0$; and (iv) $D(q, b) = 1$ for some $b \in \{0, 1\}$.
\end{proposition}

\begin{proposition}\label{proposition:weak-simulation-backwards}
  Let $D \in \pop{Q'}$ be such that $D_0 \trans{*} D$. If $D(\error) =
  0$, then there is a configuration $C$ of $\M$ such that (i) $C_0 \trans{*} C$;
  (ii) $C(x) = D(x)$ for every $x \in X$; and (iii) if $D(q, b) = 1$ for some $(q, b) \in Q'$, then $C$ is in control state $q$.
\end{proposition}

We may now prove that $\PP$ silently semi-computes $\varphi$.

\begin{proposition}\label{lemma:semicompute-nspace}
  For every $n$-bounded counter machine $\M$ that computes some
  predicate $\varphi$, there exists a broadcast consensus protocol that silently
  semi-computes $\varphi$.
\end{proposition}

\begin{proof}
  We show that $\PP$ silently semi-computes $\varphi$ by proving the
  two properties of
  Definition~\ref{definition:silently-semicompute}. Let $\vec{\alpha}$
  be an input.
  \begin{enumerate}
    \item Assume $\varphi(\vec{\alpha}) = 1$. Then $\M$ accepts
      $\vec{\alpha}$, and so there is a configuration $C$ such that
      $C_0 \trans{*} C$ and $C$ is in control state $q_a$. By
      Proposition~\ref{proposition:weak-simulation}, there exists some
      configuration $D \in \pop{Q'}$ satisfying $D_0 \trans{*} D$,
      $D(\error) = 0$ and $D(q_a, b) = 1$. Since $\M$ halts when
      reaching $q_a$, the only transition enabled at $D$ is $t_{\one,
        b}$, and its application yields a terminal configuration $D'$
      of consensus $1$. Further, every configuration reachable from
      $D_0$, where the leader is not in position $q_a$ or where some
      nonleader is in position $\error$, can be set back to $D_0$ via
      some reset transition. Therefore, every fair execution of $\PP$
      starting at $I(\vec{\alpha})=C_0$ will eventually reach $D'$.

    \item Assume $\varphi(\vec{\alpha}) = 0$. We prove by
      contradiction that no configuration $D$ reachable from $D_0$ is
      terminal. Assume the contrary. We must have $D(q_a, 1) = 1$,
      $D(\error) = 0$ and $O(D) = 1$, for otherwise some broadcast
      transition with $f_\res$ or $f_\one$ would be enabled. From this
      and by Proposition~\ref{proposition:weak-simulation-backwards},
      there exists some configuration $C$ of $\M$ in control state
      $q_a$ and satisfying $C_0 \trans{*} C$. Thus, $\M$ accepts
      $\vec{\alpha}$, contradicting
      $\varphi(\alpha)=0$. \hfill\qedhere
  \end{enumerate}
\end{proof}

\subsubsection{Main theorem}\label{subsec:main}

We prove our main result, namely that broadcast consensus protocols precisely compute the predicates in $\NL$.

\begin{theorem}
Broadcast consensus protocols compute exactly the predicates in $\NL$.
\end{theorem}

\begin{proof}
  Proposition \ref{prop:gen:upper:bound} shows that every predicate
  computable by broadcast consensus protocols is in $\NL$. For the
  other direction, let $\varphi$ be a predicate in $\NL$. Since ${\NL
  = \coNL}$ by Immerman-Stelepcs\'enyi's theorem, the complement
  predicate $\overline{\varphi}$ is also in $\NL$. Thus, $\varphi$ and
  $\overline{\varphi}$ are computable by $\bigO(\log n)$-space-bounded
  nondeterministic Turing machines. By Theorem~\ref{thm:FMR68} and
  Proposition~\ref{lem:poly:input}, $\varphi$ and $\overline{\varphi}$
  are computable by polynomially bounded counter machines, and thus by
  $n$-bounded counter machines. Therefore, by
  Proposition~\ref{lemma:semicompute-nspace}, $\varphi$ and
  $\overline{\varphi}$ are silently semi-computable by broadcast
  consensus protocols. By Proposition~\ref{prop:fromsemitocomp}, this
  implies that $\varphi$ is silently computable by a broadcast
  consensus protocol.
\end{proof}

Actually, the proof shows this slightly stronger result:

\begin{corollary}
\label{cor:silently}
A predicate is computable by a broadcast consensus protocol if{}f it
is silently computable by a broadcast consensus protocol. In
particular, broadcast consensus protocols silently compute all
predicates in $\NL$.
\end{corollary}

\section{Subclasses of broadcast consensus protocols}
While broadcasting is a natural, well understood, and much used
communication mechanism, it also consumes far more energy than
rendez-vous communication. In particular, agents able to broadcast are
more expensive to implement. In this section, we briefly analyze which
restrictions can be imposed on the broadcast model without reducing
its computational power. We show that all predicates in $\NL$ can be
computed by protocols satisfying two properties:
\begin{enumerate}
\item only one agent broadcasts; all other agents only use rendez-vous
  communication.

\item the broadcasting agent only needs to send one signal, meaning
  that the receivers' response is independent of the broadcast signal.
\end{enumerate}

Finally, we show that a third restriction \emph{does} decrease the
computational power. In simulations of the previous section,
broadcasts are often used to ``reset'' the system. Since computational
models with resets have been devoted quite some attention~\cite{LeeY94,Schnoebelen10,DufourdFS98,CHH18}, we investigate the computational power of
protocols with resets.\smallskip

\parag{Protocols with only one broadcasting agent}
Loosely speaking, a broadcast protocol with one broadcasting agent is
a broadcast protocol $\PP = (Q, R, B, \Sigma, L, I, O)$ with a set
$Q_\ell$ of leader states such that $L = \multiset{q}$ for some $q \in
Q_\ell$ (i.e., there is exactly one leader), and whose transitions
ensure that the leader always remains within $Q_\ell$, that no other
agent enters $Q_\ell$, and that only agents in $Q_\ell$ can trigger
broadcast transitions. Protocols with multiple broadcasting agents can
be simulated by protocols with one broadcasting agent, say $b$. Instead
of directly broadcasting, an agent communicates with $b$ by
rendez-vous, and delegates to $b$ the task of executing the
broadcast. More precisely, a broadcast transition $q \mapsto q'; f$ is
simulated by a rendez-vous transition ${(q, q_\ell) \mapsto (q_{aux},
q_{\ell,f})}$, followed by a broadcast transition $q_{\ell,f} \mapsto
q_\ell; (f \cup \{ q_{aux} \mapsto q' \})$.

\parag{Single-signal broadcast protocols}
In single-signal protocols the receivers' response is independent of
the broadcast signal. Formally, a broadcast protocol $(Q, R, B,
\Sigma, I, O)$ is a \emph{single-signal protocol} if there exists a
function $f \colon Q \to Q$ such that $B \subseteq Q^2 \times \{f\}$.

\begin{restatable}{proposition}{propSingleSignal}\label{prop_single_signal}
 Predicates computable by broadcast consensus protocols are also
 computable by single-signal broadcast protocols.
\end{restatable}

\begin{proof}
We give a proof sketch; details can be found in the Appendix.  We
simulate a broadcast protocol $\PP$ by a single-signal protocol
$\PP'$. The main point is to simulate a broadcast step
$C_1 \trans{q_1 \mapsto q_2; g} C_2$ of $\PP$ by a sequence of steps
of $\PP'$.

In $\PP$, an agent at state $q_1$, say $a$, moves to $q_2$, and broadcasts the signal
with meaning ``react according to $g$''. Intuitively, in $\PP'$,  agent $a$ broadcasts the
unique signal of $\PP'$, which has the meaning ``freeze''. An agent that receives the signal, say $b$,  becomes ``frozen''. Frozen agents can only be ``awoken'' by a rendez-vous with $a$. When the rendez-vous happens, $a$ tells $b$ which state it has to move to according to $g$.

The problem with this procedure is that $a$ has no way to know if it has already performed
a rendez-vous with all frozen agents. Thus, frozen agents can spontaneously
move to a state $\error$ indicating ``I am tired of waiting''. If an agent is in this state, then eventually all agents go back to their initial states, reinitializing the computation. This is achieved by   letting agents in state $\error$ move to their initial states while broadcasting the ``freeze'' signal.
\end{proof}

\parag{Protocols with reset}
In protocols with reset, all broadcasts transitions reset the protocol
to its initial configuration. Formally, a \emph{population protocol
  with reset} is a broadcast protocol $\PP = (Q, R, B, \Sigma, I, O)$
such that for every finite
execution $C_0 C_1 \cdots C_k$ from an initial configuration $C_0$, the following holds: $C_k
\trans{b} C'$ implies $C' = C_0$ for every $b \in B$ and every $C' \in
\pop{Q}$. 

\begin{restatable}{proposition}{propReset}\label{prop:reset}
  Every predicate computable by a population protocol with reset is
  Presburger-definable, and thus computable by a standard population
  protocol.
\end{restatable}

\newcommand{\bottom}{\mathcal{B}}
\newcommand{\noreset}{\mathcal{N}}

\begin{proof}
  We give a proof sketch; details can be found in the Appendix. Let
  $\PP = (Q, R, B, \Sigma, I, O)$ be a population protocol with reset
  that computes some predicate. We show that the set of accepting
  initial configurations of $\PP$, denoted $I_1$, is
  Presburger-definable as follows. Let: 
  \begin{itemize}
  \item $\PP'$ be the population protocol obtained from $\PP$ by
    eliminating the resets;
 
  \item $\noreset$ be the set of configurations $C$ of $\PP'$ from
    which no reset can occur, i.e., no configuration reachable from
    $C$ enables a reset of $\PP$;
  
  \item $S_1$ be the set of configurations $C$ of $\PP'$ that are
    stable 1-consensuses, i.e., $O(C')=1$ for every $C'$ reachable
    from $C$;

  \item $\bottom$ be the set of configurations $C$ of $\PP'$ that
    belong to a bottom strongly connected component of the
    configuration graph, i.e., $C$ can reach $C'$ iff $C'$ can reach
    $C$.  \end{itemize}
  
  We show that an initial configuration $C$ belongs to $I_1$ if{}f it
  belongs to $S_1$ or it can reach a configuration from $S_1 \cap
  \bottom \cap \noreset$. Using results
  from~\cite{esparza2017verification}, showing in particular that
  $\bottom$ is Presburger-definable, we show that $I_1$ is
  Presburger-definable.
\end{proof}

\section{Conclusion}\label{sec:conclusion}
We have studied the expressive power of broadcast consensus protocols:
an extension of population protocols with reliable broadcasts, a
standard communication primitive in concurrency and distributed
computing. We have shown that, despite their simplicity, they
precisely compute predicates from the complexity class \NL{}, and are
thus as expressive as several other proposals from the literature
which require a primitive more difficult to
implement: \emph{receiving} messages from all agents, instead of
sending messages to all agents.

As future work, we wish to study properties beyond expressiveness,
such as state complexity and space vs.\ speed trade-offs. It would
also be interesting to tackle the formal verification of broadcast
consensus protocols. Although this is challenging as it goes beyond
Presburger arithmetic and the decidability frontier, it has recently
been shown that models with broadcasts admit more tractable
approximations~\cite{BHM18}.

%% Bibliography
\bibliography{references}

%% Temporary appendix
\clearpage
\appendix

\noindent{\LARGE\textbf{\textsf{Appendix}}}

\section*{Predicates computable by broadcast consensus protocols are in $\NL$: proof of
Proposition \ref{prop:gen:upper:bound}}

\propGenUpperBound*

\begin{proof}
  Let $\ell$ be the arity of $\varphi$. We show that there is a
  nondeterministic Turing machine, that
  decides, on input $\vec{x} \in \Nat^\ell$, whether $\varphi(\vec{x})
  = 1$ holds, und uses $\bigO(\log |\vec{x}|)$ space.

  Let $d \defeq |I(\vec{x})| = |\vec{x}|$. Let $G = (V,
  E)$ be the graph where $V$ is the set of all configurations of $\PP$
  of size $d$, and $(C, C') \in E$ if{}f $C \trans{} C'$. Every
  node of $V$ can be stored using at most $\bigO(\ell \cdot \log d)$ space
  and so, since $\ell$ is fixed, in space $\bigO(\log d)$. Since $C \trans{*} C'$
  implies $|C| = |C'|$, the set $V$ contains
  all configurations reachable from $C_0$.

  We claim that $\varphi(\vec{x}) = 1$ if{}f $G$ contains a
  configuration $C$ satisfying (1) $C_0 \trans{*} C$;\label{prop:reach} and (2)
  every configuration reachable from $C$, including $C$
    itself, is a $1$-consensus.

  If $\varphi(\vec{x}) = 1$, then such a configuration $C$ exists by
  definition. For the other direction, assume some configuration $C$
  satisfies both properties. By assumption, there exists a fair
  execution starting at $C_0$ that converges to $1$. Thus, since
  $\PP$ computes $\varphi$, every fair execution starting at $C_0$
  converges to $1$, and so $\varphi(\vec{x}) = 1$.

  By the claim, it suffices to exhibit a nondeterministic Turing
  machine $\TM_{12}(C)$ that
  runs in $\bigO(\log |C|)$ space and accepts a configuration $C$
  if{}f $C$ satisfies properties~(1) and (2).

  We first observe that there is a nondeterministic Turing machine
  $\TM(C)$ that runs in $\bigO(\log |C|)$ space and accepts the
  configurations $C$ satisfying the following property:
  \begin{quote}
    there exists a configuration $C'$ such that $C \trans{*} C'$
    and $C'$ is \emph{not} a $1$-consensus.
  \end{quote}
  The machine $\TM(C)$ starts at $C$, guesses a path of
  configurations step by step, and checks that the final configuration
  $C'$ is not a $1$-consensus. While guessing the path, the machine
  only stores two configurations at any given time, and so, since
  every configuration reachable from $C$ has the same size as $C$,
  and since $\Step \in \NL$, the machine only uses $\bigO(\log
  |C|)$ space. Checking whether $C'$ is not a 1-consensus can be
  done in constant space.

  Now we use the fact that space complexity classes are closed under
  complement~\cite{Imm88}. Since $\NL = \coNL$, there
  exists a nondeterministic Turing machine $ \overline{\TM}(C)$ that,
  given as input a configuration $C$, decides in $\bigO(\log |C|)$
  space whether every configuration reachable from $C$ is a
  consensus.

  The machine $\TM_{12}(C)$ first
  guesses a configuration $C'$ reachable from $C$, proceeding as in
  the description of $\TM(C)$, and then simulates
  $\overline{\TM}(C')$. Clearly, the machine runs in $\bigO(\log
  |C|)$ space.
\end{proof}

\section*{Simulation of polynomially bounded counter machines by $n$-bounded counter machines: proof Lemma~\ref{lem:poly:input}}

\newcommand{\Xo}{\overline{X}}
\newcommand{\zn}{z_n}
\newcommand{\zt}{z_0}

Recall that a counter machine is $n$-bounded if $|C| \leq |D|$ for
every initial configuration $D$ and every configuration $C$ reachable
from $D$. We relax this definition and say that a counter machine is
\emph{weakly $n$-bounded} if the property ``$|C| \leq |D|$'' is
replaced by ``$C(x) \leq |D|$ for every counter $x \in X$''. In other
words, in the weak setting, each counter is $n$-bounded instead of
having the sum of counters $n$-bounded.

We prove Lemma~\ref{lem:poly:input} in two steps. We first show that a
polynomially bounded counter machine can be converted to an equivalent
weakly $n$-bounded machine. We then show that a weakly $n$-bounded
machine can be made $n$-bounded.

\begin{proposition}\label{prop:poly:weaklin}
  For every polynomially bounded counter machine that computes some
  predicate $\varphi$, there exists a weakly $n$-bounded counter
  machine that computes $\varphi$.
\end{proposition}

\begin{proof}
  Let $c \in \N_{> 0}$ and let $\M = (Q, X, \Delta, m, \allowbreak
  q_0, q_a, q_r)$ be an $n^c$-bounded counter machine.\medskip

  \parag{Counter values representation} To simulate $\M$ by a weakly
  $n$-bounded counter machine $\Mo$, we need some way to represent any
  counter value from $[0, n^c]$ with counters with values from $[0,
    n]$. Note that any number $\ell \in [0, n^c]$ can be encoded in
  base $n + 1$ over $c$ counters. For example, if $c = 3$ and $n = 4$,
  then $59 = 2 \cdot 5^2 + 1 \cdot 5^1 + 4 \cdot 5^0$ and hence $59$
  can be represented by $(2, 1, 4)$. Thus, we represent numbers from
  $[0, n^c]$ this way by vectors of $[0, n]^c$.\medskip

  \parag{Counters and states} We replace every counter $x$ by $c$
  counters: $\Xo_x \defeq \left\{x_i : 0 \leq i < c \right\}$, where
  $x_i$ represents digit $i$ of the base $n+1$ representation. The set
  of counters of $\Mo$ is:
  $$\Xo \defeq \bigcup_{x \in X} \Xo_x \cup \{\zt, \zn\}.$$ Machine
  $\Mo$ has the same arity as $\M$: counter $x_0$ is made an input
  counter for every $x \in X$. The control states of $\Mo$ form a
  superset of $Q$, the initial state is $\overline{q_0}$, and the
  accepting and rejecting states remain unchanged.

  Our construction will preserve the invariant $y \leq n$ for every $y
  \in \Xo$. This implies weak $n$-boundedness. Moreover, every
  forthcoming gadget will end with $\zn = n$ and $\zt = 0$.\medskip

  \parag{Transitions} Let us describe the transitions of $\Mo$. As
  depicted in Figure~\ref{fig:init}, a gadget initially computes the
  sum $n$ of the input counters, stores it in counter $\zn$, and
  restores the contents of the input counters. The purpose of $\zn$ is
  to store $n$ so that other gadgets can use it.

  \begin{figure}[!h]
    \centering
    \begin{tikzpicture}[auto, very thick, node distance=2.5cm, transform shape, scale=0.8]
      %% States
      \node[state, font=\large]              (q0p) {$\overline{q_0}$};
      \node[state]          [right of=q0p]   (0)   {};
      \node[state]          [right of=0]     (1)   {};
      \node[state]          [right=6cm of 1] (q0)  {$q_0$};

      %% Transitions
      \path[->]
      (q0p) edge node {$\zero(x_0)$} (0)

      (q0p) edge[loop above] node {
        \begin{tabular}{l}
          $\decr(x_0)$ \\
          $\incr(\zt)$ \\
          $\incr(\zn)$
        \end{tabular}
      } ()

      (0) edge[loop above] node {
        \begin{tabular}{l}
          $\incr(x_0)$ \\
          $\decr(\zt)$
        \end{tabular}
      } ()

      (0) edge node {$\zero(\zt)$} (1)
      ;

      \path[->, dotted]
      (1) edge[] node {same for other input counters} (q0)
      ;
    \end{tikzpicture}
    \caption{Gadget initializing counter $\zn$ to value $n$ and moving
      to the initial state of $\M$. Here a single transition labelled
      with $k$ instructions is a shorthand for a sequence of $k$
      transitions where the first $k-1$ transitions can be reversed.}
    \label{fig:init}
  \end{figure}
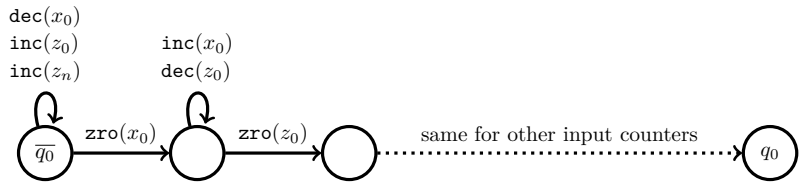

  We describe the other transitions of $\Mo$ by explaining how to
  simulate instructions $\zero(x), \nonzero(x), \incr(x)$ and
  $\decr(x)$ over $\Xo_x$, for every $x \in X$. Instruction $\nop$ is
  trivially simulated by $\nop$ itself. Let us fix $x \in X$. \medskip

  \parag{(Non)zero-tests} We have $x = 0$ if{}f $\sum_{y \in \Xo_x} y
  = 0$. Thus, instruction $\zero(x)$ is simulated by $c$ sequential
  zero-tests over $\Xo_x$. Similarly, $\nonzero(x)$ is simulated by
  parallel nonzero-tests. Both gadgets are depicted in
  Figure~\ref{fig:simul:zero}. Note that $\nop$ instructions are further added
  within the gadget for $\zero(x)$ to ensure that $\Mo$ does not block
  where $\M$ would not block.

  \begin{figure}[!h]
    \centering
    %% Gadget for zero-test
    \begin{tikzpicture}[auto, very thick, node distance=1cm, scale=0.8, transform shape]
      %% States
      \node[state, font=\large]              (q) {$q$};
      \node[state]              [right=of q] (1) {};
      \node[state]              [right=of 1] (2) {};
      \node[state]              [right=of 2] (3) {};
      \node[state, font=\large] [right=of 3] (r) {$r$};

      %% Transitions
      \path[->]
      (q) edge node[yshift=10pt, xshift=-5pt] {$\zero(x_0)$} (1)
      (1) edge node[yshift=10pt] {$\zero(x_1)$} (2)
      (3) edge node[yshift=10pt] {$\zero(x_{c-1})$} (r)

      (1) edge[bend left=30] node[xshift=7pt] {$\nop$} (q)
      (2) edge[bend left=45] node[xshift=7pt] {$\nop$} (q)
      (3) edge[bend left=60] node[xshift=7pt] {$\nop$} (q)
      ;

      \path[->, dotted]
      (2) edge node {} (3)
      ;
    \end{tikzpicture}
    \hspace*{20pt}
    %% Gadget for nonzero-test
    \begin{tikzpicture}[auto, very thick, node distance=5cm, transform shape, scale=0.8]
      %% States
      \node[state, font=\large] [below=of q]  (qn) {$q$};
      \node[state, font=\large] [right=of qn] (rn) {$r$};

      %% Transitions
      \path[->]
      (qn) edge[out=60,  in=120] node {$\nonzero(x_0)$} (rn)
      (qn) edge[out=30,  in=150] node {$\nonzero(x_1)$} (rn)
      (qn) edge[out=-30, in=210] node {$\nonzero(x_{c-1})$} (rn)
      ;

      \path[->, dotted]
      (qn) edge[out=0, in=180] node {} (rn)
      ;
    \end{tikzpicture}
    \caption{Gadgets for the simulation of instructions $(q, \zero(x),
      r) \in \Delta$ (left) and $(q, \nonzero(x), r) \in \Delta$
      (right).}
    \label{fig:simul:zero}
  \end{figure}
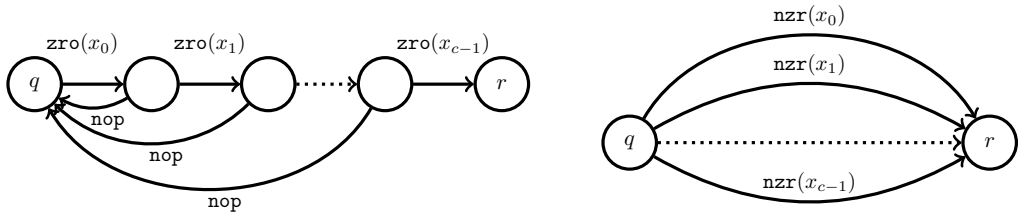

  \parag{Incrementation/decrementation} Incrementing
  (resp.\ decrementing) $x$ amounts to incrementing
  (resp.\ decrementing) its base-$(n+1)$ representation. These
  standard operations require the following primitives: ``$x_i
  \leftarrow 0$'', ``$x_i \leftarrow n$'', ``$x_i = 0$?'' and ``$x_i =
  n$?''. Their straightforward implementations are depicted in
  Figure~\ref{fig:gadgets:base}.

  \begin{figure}[!h]
    \centering
    \begin{tabular}{cp{1cm}c}
      %% Gadget for x <- 0
      \begin{tikzpicture}[auto, very thick, node distance=1.5cm, transform shape, scale=0.8]
        %% States
        \node[state, font=\large]              (q) {$q$};
        \node[state, font=\large] [right=of q] (d) {};

        %% Transitions
        \path[->]
        (q) edge[loop above] node {$\decr(x_i)$} ()
        (q) edge             node {$\zero(x_i)$} (d)
        ;
      \end{tikzpicture} &&

      %% Gadget for x <- n
      \begin{tikzpicture}[auto, very thick, node distance=1.5cm, transform shape, scale=0.8]
        %% States
        \node[state, font=\large]              (q) {$q$};
        \node[state]              [right=of q] (0) {};
        \node[state]              [right=of 0] (1) {};
        \node[state]              [right=of 1] (2) {};

        %% Transitions
        \path[->]
        (q) edge[loop above] node {$\decr(x_i)$} ()
        (q) edge             node {$\zero(x_i)$} (d)

        (0) edge[loop above] node {
          \begin{tabular}{l}
            $\incr(x_i)$ \\
            $\incr(\zt)$ \\
            $\decr(\zn)$
          \end{tabular}
        } ()

        (0) edge node {$\zero(\zn)$} (1)

        (1) edge[loop above] node {
          \begin{tabular}{l}
            $\decr(\zt)$ \\
            $\incr(\zn)$
          \end{tabular}
        } ()

        (1) edge node {$\zero(\zt)$} (2)
        ;
      \end{tikzpicture} \\[10pt]

      %% Gadget for x_i = n?
      \multicolumn{3}{c}{%
        \begin{tikzpicture}[auto, very thick, node distance=2cm, scale=0.8, transform shape]
          %% States
          \node[state, font=\large]                                (q) {$q$};
          \node[state]              [right=of q]                   (0) {};
          \node[state]              [above right=1cm and 2cm of 0] (1) {};
          \node[state]              [below right=1cm and 2cm of 0] (2) {};
          \node[state, font=\large] [right=of 1]                   (y) {yes};
          \node[state, font=\large] [right=of 2]                   (n) {no};

          %% Transitions
          \path[->]
          (q) edge[loop above] node {
            \begin{tabular}{l}
              $\decr(x_i)$ \\
              $\decr(\zn)$ \\
              $\incr(\zt)$
            \end{tabular}
          } ()

          (q) edge node       {$\zero(x_i)$}    (0)
          (0) edge node       {$\zero(\zn)$}    (1)
          (0) edge node[swap] {$\nonzero(\zn)$} (2)

          (1) edge[loop above] node {
            \begin{tabular}{l}
              $\incr(x_i)$ \\
              $\incr(\zn)$ \\
              $\decr(\zt)$
            \end{tabular}
          } ()

          (2) edge[loop above] node {
            \begin{tabular}{l}
              $\incr(x_i)$ \\
              $\incr(\zn)$ \\
              $\decr(\zt)$
            \end{tabular}
          } ()

          (1) edge node       {$\zero(\zt)$} (y)
          (2) edge node[swap] {$\zero(\zt)$} (n)
          ;
        \end{tikzpicture}
      }
    \end{tabular}
    \caption{Gadgets for primitives ``$x_i \leftarrow 0$'' (top left),
      ``$x_i \leftarrow n$'' (top right) and ``$x_i = n$?'' (bottom)
      starting from state $q$. All gadgets start and end with counter
      values $\zn = n$ and $\zt = 0$. Here a single transition
      labelled with $k$ instructions is a shorthand for a sequence of
      $k$ transitions where the first $k-1$ transitions can be
      reversed.}
    \label{fig:gadgets:base}
  \end{figure}
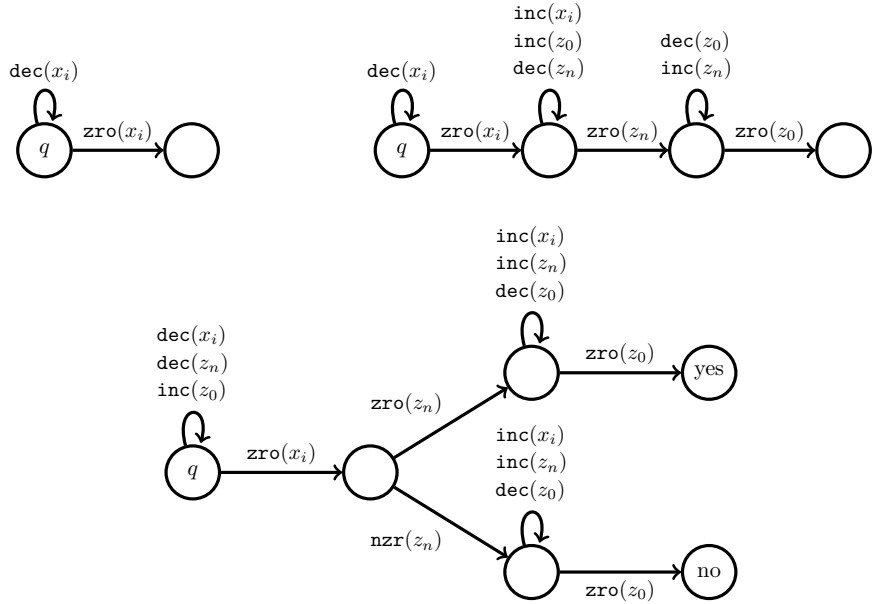
\end{proof}

\begin{proposition}\label{prop:weaklin:lin}
  For every weakly $n$-bounded counter machine that computes some
  predicate $\varphi$, there exists a $n$-bounded counter machine that
  computes $\varphi$.
\end{proposition}

\begin{proof}
  Let $\M = (Q, X, \Delta, m, \allowbreak q_0, q_a, q_r)$ be a weakly
  $n$-bounded counter machine. We construct an $n$-bounded counter
  machine $\Mo$ that simulates $\M$.\medskip

  \parag{Counter values representation} We introduce counters that
  symbolically hold values from multiple counters of $\M$. The
  counters of $\Mo$ are defined as $\Xo \defeq \{y_S : S \subseteq
  X\}$. Intuitively, if counter $y_S$ has value $a$, then it
  contributes by $a$ to the value of each counter of $S$. More
  formally, let $C$ be a configuration of $\Mo$. The value represented
  by $C$ for counter $x \in X$ is defined as:
  $$C(x) \defeq \sum_{\{x\} \subseteq S \subseteq X} C(y_S).$$

  For example, consider the case with $3$ counters and an input $(4,
  0, 2)$ of size $n = 6$. Suppose we can reach $(6, 1, 4)$ in
  $\M$. Observe that it is weakly $n$-bounded, but \emph{not}
  $n$-bounded. However, counter values $(x_1, x_2, x_3) = (6, 1, 4)$
  can be represented in $\Mo$ as:
  \begin{align*}
    y_{\{x_1, x_2, x_3\}} &= 1, \\
    y_{\{x_1, x_3\}}      &= 3, \\
    y_{\{x_1\}}           &= 2, \\
    y_S                   &= 0~\text{for every other $S$}.
  \end{align*}
  Under such a representation, the sum of all counters equals $n$,
  e.g.\ $n = 6$ here. Moreover, all counter values from $\{0, 1,
  \ldots, n\}^k$ can be represented in this fashion (in possibly many
  ways). In particular, $\vec{0}$ is represented by $y_\emptyset = n$
  and $y_S = 0$ for every $S \neq \emptyset$.\medskip

  \parag{Counters and states} Machine $\Mo$ has the same arity as
  $\M$: counter $y_{\{x\}}$ is made an input counter for every $x \in
  X$. The control states of $\Mo$ form a superset of $Q$, and the
  initial, accepting and rejecting states remain unchanged.

  Our construction will preserve the invariant $\sum_{S \subseteq X}
  y_S = n$. This implies $n$-boundedness. We define the transitions of
  $\Mo$ by describing how each instruction is implemented.\medskip

  \parag{(Non)zero-tests} We have $x = 0$ if{}f $\sum_{\{x\} \subseteq
    S \subseteq X} y_S = 0$. Thus, instruction $\zero(x)$ is simulated
  by $2^{|X|-1}$ sequential zero-tests. Similarly, $\nonzero(x)$ is
  simulated by parallel nonzero-tests. Both gadgets are analogous to
  those of Figure~\ref{fig:simul:zero} constructed for
  Proposition~\ref{prop:poly:weaklin}.\medskip

  \parag{Incrementation/decrementation} In order to increment $x \in
  X$, we nondeterministically decrement any counter $y_S$ such that
  $x \not\in S$, and increment $y_{S \cup \{x\}}$. Similarly, to
  decrement $x \in X$, we non deterministically decrement any counter
  $y_S$ such that $x \in S$, and increment $y_{S \setminus
    \{x\}}$. Note that both gadgets decrement before incrementing,
  thereby preserving $n$-boundedness.
\end{proof}

\lemPolyInput*

\begin{proof}
  This follows from applying Proposition~\ref{prop:poly:weaklin}
  followed by Proposition~\ref{prop:weaklin:lin}.
\end{proof}

\section*{Single-signal broadcast protocols: proof of Proposition~\ref{prop_single_signal}}

\propSingleSignal*

\begin{proof}
By Corollary~\ref{cor:silently}, it suffices to show that for every broadcast protocol
$\PP = (Q, R, B, \Sigma, \allowbreak L, I, O)$ that silently computes a predicate there is a single-signal protocol
$\PP' = (Q', R', B', \Sigma, \allowbreak I', O')$ that computes the same predicate.

\parag{States and mappings}
The states of $\PP'$ are defined as:
$$Q' \defeq Q^2 \cup (Q^2 \times B) \cup
  (Q^2 \times \{\textit{frozen}, \error, \textit{reset}\}).
  $$

Every state of $Q'$ has two or three components. As in the construction of Lemma~\ref{lemma:semi-to-nonsemi}, the first two components describe the \emph{position} of the agent, and its \emph{origin}. Agents never change ther origin, and so they know which state to return to if they are told to reset.

The third component can be either a broadcast transition, or one of $\{\textit{frozen}, \error, \textit{reset}\}$.
The intended meaning of agent $a$ being in a state with third component $c$ is as follows:
\begin{itemize}
\item $c \in B$: agent $a$ is in charge of simulating $c$ by first freezing all other agents and then performing a
rendez-vous with each of them;

\item $c = \textit{frozen}$: agent $a$ is currently waiting for a rendez-vous with the broadcasting agent that told it to freeze;

\item $c = \error$: agent $a$ has decided not to wait any longer for the rendez-vous with the broadcasting agent, or a new broadcast signal has been sent before the simulation of the previous broadcast is completed;

\item $c = \textit{reset}$: agent $a$ is in charge of telling other agents to reset.
\end{itemize}

The input and output mappings are defined respectively by
\begin{align*}
  I'(x) &\defeq (I(x), I(x)) \text{ for every } x, \\
  O' &\defeq O \circ \textit{pos},
\end{align*}
where $\textit{pos}$ is the function that maps every state from $Q'$ to its position, i.e. $\textit{pos}(\vec{q}) \defeq q_1$ for every $\vec{q}=(q_1, q_2) \in Q^2$ and every $\vec{q}=(q_1, q_2, x) \in Q' \setminus Q^2$.

We now describe how $\PP'$ simulates $\PP$.

\parag{Initiation of a broadcast simulation} The ``freeze'' signal used by the protocol is described by the following function $f \colon Q' \to Q'$:
  \begin{align*}
    f(\vec{q}) \defeq \begin{cases}
      (\vec{q}, \textit{frozen}) &  \text{ if } \vec{q} \in Q^2, \\
      (\vec{r}, \error) & \text{ if } \vec{q}=(\vec{r}, x) \in Q' \setminus Q^2.
    \end{cases}
  \end{align*}
Intuitively, when an agent broadcasts this signal, it tells all agents in ``normal'' states to freeze.
In an error-free simulation, all agents are in such normal states, and so the broadcasts sends all other
agents to a ``frozen state'' $(\vec{r}, \textit{frozen})$.

When an already ``frozen'' agent receives the freeze signal, this means that the simulation of a broadcast or reset is not completed before another broadcast or reset is initiated, and the frozen agent assumes an ``error state''.
An error state indicates that the population must be reset at some point in the future.

For every broadcast transition $t \colon q \mapsto q'; g$ from $B$ and every $r \in Q$, we add to $B'$ the broadcast transition:
$${(q, r) \mapsto (q', r, t)}; f.$$
The agent in state $(q', r, t)$ is in charge of simulating the effect of the broadcast transition $t$ via rendez-vous with the other agents.

\parag{Simulation of rendez-vous}
For every rendez-vous transition $(q_1, q_2) \mapsto (q_1', q_2') \in R$ and every
$r_1, r_2 \in Q$, we add to $R'$ the rendez-vous transition:
$$\big((q_1, r_1), (q_2, r_2)\big) \mapsto \big((q_1', r_1),
    (q_2', r_2) \big) $$

\parag{Receiver's response}
For every broadcast transition $t \colon q \mapsto q'; g \in B$ and every $r_1, r_2, q_2 \in Q$, we add to
$R'$ the transition:
$$\big((q', r_1, t), (q_2, r_2,\textit{frozen})\big) \mapsto \big((q', r_1, t), (g(q_2), r_2)\big).$$

\parag{Completion of a broadcast}
For every broadcast transition $t \colon q \mapsto q'; g \in B$ and every $r_1, q_2, r_2 \in Q$, we add to $R'$ the transition:
$$\big( (q', r_1, t), (q_2, r_2) \big) \mapsto ((q', r_1), (q_2, r_2) ).$$

The transitions defined thus far would suffice if broadcasts were always simulated correctly. But we cannot rule out the initiation of a broadcast simulation before a previous simulation is completed, which yields an agent in an error state. Additional transitions are thus needed for error handling. Whenever an agent is in an error state, the population must be eventually reset to start a new, clean simulation attempt. In the implementation of the reset we must ensure that ``illegitimate'' agents, that have not yet been reset, cannot interact with ``legitimate'' agents, that have already been reset.

\parag{Initiation of a reset}
Agents in an error state may transition to a reset state in order to initiate
a reset. The agent in a reset state is in charge of implementing the reset via
rendez-vous with the other agents. When a reset is initiated,
there should be \emph{precisely one} agent in a reset state, while all other agents are
temporarily \emph{disabled}, for otherwise some agent's state could be modified
by some other ``illegitimate'' agent before the reset is completed, and a reset
agent would have no means to distinguish between ``legitimate'' agents and ``illegitimate'' agents.

We implement the initiation of a reset by a broadcast.
For every $(q, r) \in Q^2$, we add to $B'$ the broadcast transition:
$${(q, r, \error) \mapsto (r, r, \textit{reset})}; f.$$

\parag{Reset to origin}
For every $\vec{q} \in Q^2$ and every $(q, r, x) \in Q' \setminus Q^2$, we add to $R'$ the transition:
$$\big( (\vec{q}, \textit{reset}), (q, r, x) \big) \mapsto \big( (\vec{q}, \textit{reset}), (r, r) \big).$$

\parag{Completion of a reset}
A reset is (perhaps prematurely) completed when an agent in a reset state resets itself.
For every $\vec{q} \in Q^2, x \in Q'$, we add to $R'$ the transition:
$$\big( (\vec{q}, \textit{reset}), x \big) \mapsto (\vec{q},x ).$$

\parag{From frozen to error}
It may be the case that a reset agent resets itself to origin before all frozen agents have been reached.
To avoid that frozen agents wait forever to be ``unfrozen'', frozen agents can non-deterministically decide to assume an error state, thereby initiating a new reset.

For every $\vec{q} \in Q^2$ and every $x \in Q'$, we add to $R'$ the transition:
$$\big( (\vec{q}, \textit{frozen}), x \big) \mapsto \big( (\vec{q},
    \error), x \big).$$

  $\PP'$ computes the same predicate as $\PP$: Since $\PP$ is silent,
  every fair execution of $\PP$ reaches a terminal configuration, and by
  construction of $\PP'$, every correct simulation of $\PP$ in $\PP'$
  eventually reaches a terminal configuration of the same consensus.

  An error in the simulation may occur in one of two cases:
  Either another broadcast signal is initiated before the simulation of the
  last broadcast or a reset is completed, or the agent in charge of a reset reverts
  to its initial state before all other agents have been reset.
  In the former case, at least one agent is sent to an error state, which
  will eventually lead to a reset. In the latter case, at least one agent
  in an error state or a frozen agent remains. Frozen agents can non-deterministically
  choose to turn to error states, and thus eventually initiate a reset.
  In either case, the population is eventually reset to its initial configuration.
  Fairness guarantees that the simulation is eventually executed correctly.

  Note that silentness of protocol $\PP$
  is crucial for the correctness of the construction:
  Since $\PP$ is silent,  broadcast signals are bound to cease to occur in every
  correct simulation, and thus all agents eventually remain unfrozen forever, hence we may
  safely demand that frozen agents non-deterministically turn to error states, which
  allows us to handle incomplete resets.
\end{proof}

\section*{Protocols with reset: proof of Proposition \ref{prop:reset}}

\propReset*

\begin{proof}
  Let $\PP = (Q, R, B, \Sigma, I, O)$ be a population protocol with
  reset that computes some predicate. Let $\trans{*}$ be the reflexive
  and transitive closure of the relation $\trans{}$ of $\PP' \defeq
  (Q, R, \Sigma, I, O)$. For every $X \subseteq \pop{Q}$, let
  $$
  \pre^*(X) \defeq \{q \in \pop{Q} : q \trans{*} r\ \text{for some}\ r \in X\}.
  $$ Let $\mtrans{*}\ \defeq \{(C, C') \in \pop{Q}^2 : C \trans{*}
  C' \trans{*} C\}$ denote the \emph{mutual-reachability relation} of
  $\PP'$. A \emph{bottom strongly connected component (BSCC)} is a non
  empty set of pairwise mutually reachable configurations closed
  under reachability. We call a configuration $C$ \emph{bottom} if
  $C \in X$ for some BSCC $X$. Further let $\bottom$ denote the set of
  all bottom configurations of $\PP'$.  Recall
  from~\cite{esparza2017verification} that an execution $C_0 C_1
  C_2 \ldots$ of $\PP'$ is fair if and only if $C_i \in \bottom$ for
  all but finitely many indices~$i$.

  Let $\noreset$ be the set of configurations of $\PP'$ from which no
  reset can occur, i.e.\ let
  $$\noreset \defeq \{C \in \pop{Q} : \forall C' \in \pop{Q}, \forall
    t \in B \colon \text{if}\ C \trans{*} C', \text{then $t$ is disabled
      in $C'$}\}$$
  For every $b \in \{0, 1\}$, let $F_b \defeq S_b \cap \bottom \cap
  \noreset$, where $S_b$ is the set of \emph{$b$-stable}
  configurations: $$S_b \defeq \{C \in \pop{Q} : \forall C',\ \text{if}\ C
  \trans{*} C', \text{then}\ O(C') = b \}.$$

  We claim that the set of accepting initial configurations of $\PP$,
  denoted $I_1$, equals $\pop{I} \cap (S_1 \cup \pre^*(F_1))$. Let us
  prove the claim. Let $C_0$ be from the latter set. Either $C_0 \in
  S_1$ or $C_0 \in \pre^*(F_1)$. If $C_0 \in S_1$, then $C_0$ is
  $1$-stable in $\PP'$. Notice that every $1$-stable initial
  configuration of $\PP'$ is also $1$-stable in $\PP$, and
  consequently $C_0$ is accepting. If $C_0 \in \pre^*(F_1)$, then by
  definition of $F_1$, there is a BSCC $X \subseteq \bottom$ reachable
  from $C_0$ such that $O(C) = 1$ and such that resets are disabled
  for every $C \in X$. Since no reset is enabled in $X$, set $X$ is a
  BSCC not just in $\PP'$, but also in $\PP$. Hence, at least one fair
  execution of $\PP$, starting in $C_0$, stabilizes to $1$ and thus
  all fair executions of $\PP$ starting in $C_0$ stabilize to $1$. The
  converse direction is proven analogously.

  It remains to show that $I_1$ is Presburger-definable.
  For this, we
  make use of the following results
  from~\cite{esparza2017verification}:
  \begin{itemize}
  \item $S_0$, $S_1$, $\mtrans{*}$ and $\bottom$ are
    Presburger-definable;
  \item for every Presburger-definable sets $X, F_0, F_1 \subseteq \pop{Q}$,
    if sets $X \cap \pre^*(F_0)$ and $X \cap \pre^*(F_1)$ form a partition of $X$, then both sets
    are Presburger-definable.
  \end{itemize}

  By the above, $I' \defeq \pop{I} \setminus (S_0 \cup S_1)$ is a
  boolean combination of Presburger-definable sets, and hence
  Presburger-definable too. Similarly, $F_b$ is Presburger-definable
  for every $b \in \{0, 1\}$ through the Presburger formula
  $\psi_b(C)$:
  $$\psi_b(C) \defeq (C \in \bottom) \land \forall C' : \hspace{-11pt}
  \bigwedge_{(q, q', f) \in B} \hspace{-10pt} [(C \mtrans{*} C')
    \implies (C'(q) = 0)].$$ Since $F_0, F_1$ and $I'$ are
  Presburger-definable, and since sets $I'_0 \defeq I' \cap
  \pre^*(F_0)$ and $I'_1 \defeq I' \cap \pre^*(F_1)$ form a partition
  of $I'$, it follows by the above observations that both $I'_0$ and
  $I'_1$ are Presburger-definable.  Moreover, it is relatively
  straightforward to see that the following equalities hold:
  \begin{align*}
    I_1
    &= \pop{I} \cap (S_1 \cup \pre^*(F_1)) \\
    %&= (\pop{I} \cap S_1) \cup (\pop{I} \cap \pre^*(F_1)) \\
    &= \pop{I} \cap (S_1 \cup (I' \cap \pre^*(F_1))) \\
    &= \pop{I} \cap (S_1 \cup I'_1).
  \end{align*}
  Therefore, $I_1$ is Presburger-definable, as it is a boolean
  combination of Presburger-definable sets.
\end{proof}

\end{document}